\documentclass[journal]{IEEEtran}

\usepackage{inputenc}
\usepackage{mathtools}
\usepackage{color}
\usepackage{array}
\usepackage{verbatim}
\usepackage{float}
\usepackage{amsmath}
\usepackage{amsthm}
\usepackage{amssymb}
\usepackage{graphicx}
\usepackage{longtable}
\usepackage{multirow}
\usepackage{balance}

\usepackage[unicode=true,
bookmarks=false,
breaklinks=false,pdfborder={0 0 1},colorlinks=false]
{hyperref}
\hypersetup{
	colorlinks,bookmarksopen,bookmarksnumbered,citecolor=red,urlcolor=red}
\usepackage{cite}

\usepackage{indentfirst}
\usepackage{textcomp}
\usepackage[T1]{fontenc}
\usepackage{amsthm}
\usepackage{amsmath,amstext,amsfonts}
\usepackage{amssymb}
\usepackage{mathtools}
\usepackage[cmintegrals]{newtxmath}
\usepackage{bm}  % bold math anywhere in math mode and of greek letters too
\usepackage{siunitx}
\usepackage{physics}
\DeclareMathOperator*{\argmin}{arg\,min}

\newtheorem{lemma}{Lemma}

\usepackage{graphicx}
\graphicspath{%
	{Figures/}
	{Figures/TiKz/}
	{Figures/Ipe/}
	{Figures/XFig/}
	{Figures/Inkscape/}
	{Figures/Matlab/}
	{Figures/Scanned/}
	{Figures/Saved/}
}
\usepackage{tikz}
\usepackage[list=true,labelformat=simple]{subcaption}
\usepackage{cite}
\usepackage{url}
\usepackage{xcolor}
\usepackage{siunitx}

\usepackage[inline]{enumitem}

\usepackage{algorithm,algpseudocode,float}
\usepackage{setspace} % Needed to fix line spacing inside the algorithms

\algnewcommand{\algorithmicand}{\textbf{ and }}
\algnewcommand{\algorithmicor}{\textbf{ or }}
\algnewcommand{\AlgAnd}{\algorithmicand}
\algnewcommand{\AlgOr}{\algorithmicor}

\renewcommand{\qedsymbol}{$\blacksquare$}
\usepackage[noabbrev]{cleveref}  % automatically sort references
\Crefname{figure}{Fig.}{Figs.}

\usepackage{tabularx,booktabs}
\newcolumntype{C}{>{\centering\arraybackslash}X} % centered version of "X" type
\usepackage{lipsum}

\newcolumntype{b}{X}
\newcolumntype{s}{>{\hsize=.5\hsize}X}

\makeatletter
\let\oldforeign@language\foreign@language
\DeclareRobustCommand{\foreign@language}[1]{%
	\lowercase{\oldforeign@language{#1}}}

\renewcommand{\qedsymbol}{$\blacksquare$}

\floatstyle{ruled}
\newfloat{algorithm}{tbp}{loa}
\providecommand{\algorithmname}{Algorithm}
\floatname{algorithm}{\protect\algorithmname}

\let\oldforeign@language\foreign@language
\DeclareRobustCommand{\foreign@language}[1]{%
	\lowercase{\oldforeign@language{#1}}}

\ifCLASSINFOpdf
\else
\fi

\newtheorem{thm}{Theorem}
\newtheorem{rem}{Remark}

\newtheorem{assum}{Assumption}
\begin{document}
	
	
	\title{An Observer-Based Reinforcement Learning Solution for Model-Following Problems}

	% % % To be REMOVED
	\author{Mohammed I. Abouheaf, Kyriakos G. Vamvoudakis, Mohammad A. Mayyas, and Hashim A. Hashim % <-this % stops a space
		%\thanks{This work was supported in part by Thompson Rivers University Internal research fund, RGS-2020/21 IRF, \# 102315.}
		\thanks{
	This work was supported in part by National Science Foundation under grant Nos. S\&AS-1849264, CPS-1851588, and CPS-2038589 and by National Sciences and Engineering
	Research Council of Canada (NSERC), under the grants RGPIN-2022-04937. 
}
		\thanks{M. I. Abouheaf and M. A. Mayyas are with the Robotics Engineering, Bowling Green State University, Bowling Green, OH, 43403, USA, email: \{mabouhe,mmayyas\}@bgsu.edu.
			K. G. Vamvoudakis is with the Daniel Guggenheim 	School of Aerospace Engineering, Georgia Institute of Technology, Atlanta, GA, 30332, USA, e-mail: kyriakos@gatech.edu.
			H. A. Hashim is with the Department of Mechanical and Aerospace Engineering, Carleton University, Ottawa, Ontario, K1S-5B6, Canada, e-mail: HashimMohamed@cunet.carleton.ca.
}
}
	
	\maketitle
	\pagestyle{empty}
	\thispagestyle{empty}
	
	\begin{abstract}
In this paper, a multi-objective model-following control problem is solved using an observer-based adaptive learning scheme. The overall goal is to regulate the model-following error dynamics along with optimizing the dynamic variables of a process in a model-free fashion. This solution employs an integral reinforcement learning approach to adapt three strategies. The first strategy observes the states of desired process dynamics, while the second one stabilizes and optimizes the closed-loop system. The third strategy allows the process to follow a desired reference-trajectory. The adaptive learning scheme is implemented using an approximate projection estimation approach under mild conditions about the learning parameters. 
	\end{abstract}
	
	\IEEEpeerreviewmaketitle{}

\section{Introduction}
Model-following techniques have been adopted to find solutions for the trajectory-tracking control problems. This category of problems has been tackled using the optimal tracking control framework~\cite{Lewis2012}. The solution is accomplished by solving a set of coupled differential equations offline and then apply the derived strategy forwards in time. Such solutions enforced usage of full information of the process dynamics. On a relevant side, the adaptive approaches have been employed to solve reference-tracking control problems in real-time~\cite{aastrom2013adaptive}. Nonetheless, this category of solutions is associated with some limitations. The underlying trajectory-tracking approaches i) regulate the tracking error dynamics without optimizing the dynamic variables of the process itself and ii) are either partially or fully dependent on the process dynamics. On another side, the observer strategy requires existence of exact or approximate model of the process. Further, the observation gain does not employ flexible-order of observation error dynamics and relies on a model-based low-order error structure. These limitations motivate an innovative model-free control architecture that adopts an observer idea to solve the model-following problem for Linear Time Invariant (LTI) systems. This solution employs an Integral Reinforcement Learning (IRL) approach under mild conditions about the learning parameters to guarantee convergence.

The model-following applications involve the regulation of hypersonic aircraft, autonomous vehicles, underactuated systems, and robotic manipulators~\cite{Byrne1995,Kam1,Liu2018,Chen2021,MPC2021,Hol2002}.     
The aforementioned model-following limitations are observed in several solutions such as dual mode predictive control~\cite{DMMPC2021}, gain scheduling~\cite{Kam1}, sliding mode surfaces~\cite{AUV1,Robust2019}, adaptive backstepping \cite{hashim2023observer,hashim2023exponentially} and $\mathcal{L}_1$ adaptive control~\cite{Chen2021},  Lyapunov theory-based MRAS~\cite{Byrne1995,Moore2014}, model predictive control~\cite{MPC2014,MPC2021}, barrier function-based MRAS~\cite{Liu2018}, means of linear matrix inequalities~\cite{SHI2017}, and feedforward control~\cite{netw1}.
The graphical games have been utilized to solve leader-follower control problems for LTI agents interacting using graph topologies~\cite{AbouheafCTT2015,AbouheafICRA19,AbouheafAuto14}. These relied on pinning control ideas to ensure synchronization among the agents.
A model-based approach adopted the sum-of-squares polynomial notion to solve nonlinear model-following control problems in~\cite{Moore2014}. Another robust model-following approach adopted a model-predicted control concept for nonlinear systems~\cite{RMPC2021}. 
Further, a sliding surface-based observer scheme is employed to solve a model-following control problem for nonlinear systems~\cite{Robust2019}. It mandated existence of partial information of the process dynamics along with a zero-state detectability restriction to ensure stability. 
A similar conclusion is drawn when an observer-based  model-reference adaptive system is adopted to control sensor-less brushless doubly-fed induction machine~\cite{Obs2sen}. Moreover, a fuzzy-state observer is employed to actuate a Maglev grasping robot arm~\cite{Obsv_2023}. These observer-based schemes do not produce model-free strategies. Herein, a new observer idea is developed using a model-free strategy.
Further, an adaptive learning scheme based on IRL will be considered to solve the model-following control problem.

Reinforcement Learning (RL) is a machine learning tool that makes use of temporal difference structures to search for the optimal strategy-to-follow in a dynamic learning environment~\cite{sut92,Sutton_1998,Bertsekas1996,abouheaf2023online}. This results in a series of penalization or rewards aiming to maximize the cumulative sum of rewards. The RL solutions are realized using two-step techniques such as Value Iteration (VI) and Policy Iteration (PI)~\cite{Bertsekas1996,AbouheafACC19,Busoniu2010}. Hence, parameter estimation approaches are adopted to find the underlying strategies such as Recursive Least Squares (RLS), Batch Least Squares (BLS), and other approximate estimation techniques, to name a few~\cite{Busoniu2010, Srivastava2019}.
The optimal control mathematical setup in the continuous-time domain results in temporal difference form, namely the integral Bellman optimality equation~\cite{IRL_Kyr}. This equation can be solved in continuous-time mode using the means of IRL approaches.
The Bellman or integral Bellman optimality equations cannot be solved analytically. Hence, the adaptive critic structures are needed to approximate the RL solutions~\cite{Sutton2008,CrtBahare,CrtLewis,CrtZhao}. The adaptation mechanisms for such structures rely on gradient approaches to regulate the tuning errors. The RL techniques have been adopted to solve different problems such as the Linear Quadratic Regulator (LQR)~\cite{Bah2017}, output-based regulation of multi-agent systems~\cite{Lewis20}, model-following control~\cite{MFAC2021,Bahare14}, and control of flexible wing aircraft~\cite{AbouheafTrans20}. In this work, an IRL approach will be adopted to solve the model-following control problem.

\paragraph*{Contributions} A customized control structure is introduced to solve the model-following control problem. It comprises three model-free strategies that are able to i) regulate the model-following tracking errors, ii) regulate the observation errors, and iii) optimize the closed-loop performance of the process. This scheme does not employ the dynamics of the process explicitly in any of the developed model-free strategies. Further, the observer strategy relies on a flexible-order of the error dynamics, rather than depending on a low-order scheme. The observer itself constitutes an additional model-following structure that aims to guide the internal dynamics of the process. Unlike many approaches, the proposed one is able to optimize not only the model-following error dynamics but also the closed-loop dynamic performance. This is done by solving the underlying Linear Quadratic Regulation (LQR) problem without involving the exact or approximate dynamics of the process in the designated strategies.

\paragraph*{Mathematical notation} In this paper $\mathbb{R}$ refers to the set of real numbers. The non-negative integers and positive whole numbers are denoted by $\mathbb{Z}_{0}^{+}$ and $\mathbb{N}$, respectively. The Kronecker product is symbolized by {\tiny $\otimes$}. The gradient of function $\mathcal{M} $ is referred to as $\grad\mathcal{M}$. Let $ ||\varkappa ||_{\infty} = \sup\limits_{k \in \mathbb{N}} || \varkappa(k) ||_\infty$ define the $\mathcal{L}_\infty-$ norm of a sequence $\{ {\varkappa(k)} \}_{k=0}^{\infty}$ with $\mathcal{L}_2 \{\varkappa: ||\varkappa ||_{2} < \infty\}$ and $\mathcal{L}_\infty \{\varkappa: ||\varkappa ||_{\infty} < \infty\}$.

\paragraph*{Structure} The paper is organized as follows: 	
Section~\ref{sec:Preliminaries} details the overall control scheme comprised of three model-free strategies used to achieve the optimization goals. The optimal control setup leading to a temporal difference form, namely the integral Bellman optimality equation, is discussed in Section~\ref{sec:OPT_IBE}. Moreover, this Section outlines the stability characteristics of the solution. Section~\ref{sec:IRL_Sol} presents the model-free IRL solution and its actor-critic implementation. Furthermore, it introduces an approximate projection technique to tune the actor-critic weights to ensure stable adaptations. Section~\ref{sec:simulation} validates the IRL solution using unstable dynamic process and a nonlinear reference-trajectory. Finally, Section~\ref{sec:conclus} highlights the main findings.

\section{Problem Formulation\label{sec:Preliminaries}}
The model-following problem is faced by the complexity of the mathematical manipulations of the reference-tracking error dynamics. Further, it overlooks the simultaneous optimization of the remaining dynamic variables of the process which could cause an additional computational burden. In the sequel, the overall control scheme is explained using an observer-based strategy. The process dynamics structure is given by
\begin{equation}
	\dot {\bf \mathcal{X}}={\bf A \, \mathcal{X}}+{\bf B \, u} \quad \text{and} \quad Y=C\, \mathcal{X},
	\label{eq:dyn}
\end{equation}
where ${\bf \mathcal{X}} \in \mathbb{R}^{n},$ ${\bf \, u} \in \mathbb{R}^{m},$ and $Y\in\mathbb{R}^{p}$ are vectors of the states, input control signals, and output signals, respectively. Further, ${\bf A},$ ${\bf B},$ and ${\bf C}$ are the dynamic parameters of the process. 

The process~\eqref{eq:dyn} is required to follow another dynamical model described by
\begin{equation*}
 \dot{\hat{\bf\mathcal{X}}} \, = \, {\bf\hat  A} \, {\bf \,   \mathcal{\hat X}} \, + \,{\bf \hat B \, \left({\bf u}^{\pi_{Ob}}+{\bf u}\right)},
\label{eq:dyna}
\end{equation*}
where ${\bf \,   \mathcal{\hat X}} \in \mathbb{R}^{n}$ is a vector of the desired or observed states, ${\bf\hat A}$ and ${\bf\hat  B}$ are parameters of the desired dynamical system or approximated process, and ${\bf u}^{\pi_{Ob}}\in \mathbb{R}^{m}$ is the control signal due to an observer strategy ${\pi_{Ob}}$.
\begin{rem}
Parameters ${\bf\hat A}$ and ${\bf\hat  B}$ could serve as the desired dynamical performance required from the process or could even represent some approximation of that process (i.e., ${\bf A}$ and ${\bf B}$). The developed adaptive learning mechanism will not employ the dynamic parameters of the process nor its counterpart to provide a solution. The goal is to regulate the observation errors without either having partial or full knowledge of the process dynamics. 
\end{rem} 

The reference-trajectory dynamical behavior is prescribed by an independent command generator given by

%	\label{eq:ref}
${\bf Y}^{ref}(t)=f(t), \,$ where ${\bf Y}^{ref}(t) \in \mathbb{R}^{q}$.

The control solution will adopt the form in~\eqref{eq:dyn} to generate the online measurements, while the dynamic system parameters will not be explicitly employed in any part of the strategies. 
The overall objective of the optimization problem is to let an output ${\bf Y}^s \in \mathbb{R}^{q}$ of system~\eqref{eq:dyn} follow the reference-trajectory ${\bf Y}^{ref}(t) \in \mathbb{R}^{q}$ (i.e., $\lim\limits_{t \rightarrow \infty}\lVert{{\bf e}^{Mf}(t)}\rVert\rightarrow \bf 0,$ ${\bf e}^{Mf}(t)={\bf Y}^{ref}(t)-{\bf Y}^s(t)$). However, many of the available solutions can work either offline or use model-based approaches. This requires analytical manipulation of the error dynamics. Further, the regulation of the remaining dynamic variables is overlooked while following the desired trajectory. For the clarity of mathematical notation, any time-dependent function $g(t)$ will be referred to as $g_t$.

In order to tackle the raised challenges, the overall control strategy is divided into three sub-strategies: (i) one strategy ${\bf u}_t^{\pi_{Ob}}\in \mathbb{R}^{m}$ observes the states of the process. It represents an additional model-following loop that compares the outputs of the process~\eqref{eq:dyn} to those of the desired or approximated ones, (ii) another strategy ${\bf \mu}_t^{\pi_{Cl}}\in \mathbb{R}^{m}$ optimizes the closed-loop performance of the dynamic system, and (iii) the third strategy ${\bf u}_t^{\pi_{Mf}}\in \mathbb{R}^{m}$ reflects the model-following actions. These interactive strategies are implemented in a model-free fashion. Hence, the main control strategy can be written such that ${\bf u}_t={\bf \mu}_t^{\pi_{Cl}}+{\bf u}_t^{\pi_{Mf}}$. In the sequel, the detailed control scheme will be explained. 
\subsection{Observing the Desired Dynamic Performance}  
This strategy aims to find the desired states using an observer-like structure. Hence, the desired or approximated dynamic process is described by 
\begin{equation}
\dot {\bf \hat{\mathcal{X}}}\,=\,{\bf \hat A \, \mathcal{\hat X}}\,+\,{\bf \hat B} \left({\bf u}^{\pi_{Ob}}+{\bf u}\right) \, \, \text{and}\, \, \hat Y \,=\, C\,  \mathcal{\hat X},
\label{eq:dyna}
\end{equation}
where ${\bf \mathcal{\hat X}} \in \mathbb{R}^{n}$ and $\hat Y\in\mathbb{R}^{p}$ are vectors of the desired or observed states and output signals, respectively.

The observer control signal ${\bf u}^{\pi_{Ob}}$ relies on a flexible-order of tracking-error dynamics which is dedicated by a vector ${\bf E}^{Ob}$. The size of this vector varies according to the number of error samples $e^{Ob}_t={\bf Y}_t-{\bf \hat Y}_t$ collected at a fixed-time interval $\delta$ such that ${\bf E}^{Ob}_t= \left[ \begin{array}{ccc}
e^{Ob}_t & e^{Ob}_{t+\delta} & e^{Ob}_{t+2\delta} \end{array}\right]^\textrm{T}  \in \mathbb{R}^{3p}$. The observer strategy ${\pi_{Ob}}$ is selected using an adaptive learning mechanism and the resulting control signal is given by ${\bf u}^{\pi_{Ob}}_{t+\delta}={\bf u}^{\pi_{Ob}}_t+\mu_t^{\pi_{Ob}},$ where $\mu_t^{\pi_{Ob}}={\bf \pi}_{Ob}\,{\bf E}^{Ob}_t, \, \mu_t^{\pi_{Ob}} \in \mathbb{R}^{m}$. The strategy ${\bf \pi}_{Ob}$ is selected to minimize the following performance index
\begin{equation}
	J_t^{\bf \pi_{Ob}}=\int_t^\infty {U}_\tau^{Ob} \, \left({\bf E}^{Ob}_\tau,\mu_\tau^{\pi_{Ob}}\right)\,d\tau,
	\label{eq:index_Ob}
\end{equation} 
where ${U}^{Ob}$ is an objective cost function to minimize the cumulative observation errors.
\begin{assum}
	\label{Asm:1}
	The dynamic system~\eqref{eq:dyna} defined by $\bf (\hat A,
	\, C)$ is observable and the process~\eqref{eq:dyn} is observable as well. \hfill $\square$
\end{assum}
This observer strategy ${\pi_{Ob}}$ will be determined in a model-free fashion using the adaptive learning approach. 
\subsection{Closed-Loop Strategy}  
The model-following strategy regulates the trajectory-tracking error dynamics while stabilizing and optimizing the performance of the process. Herein, a closed-loop feedback strategy ${\bf\pi}_{Cl}$ will be advised based on the observed states in real-time. This is to ensure stability of the closed-loop dynamical system, provided that this system is stabilizable. Hence, the objective function associated with this strategy is given by
\begin{equation}
	J_t^{\bf \pi_{Cl}}=\int_t^\infty {U}_\tau^{Cl} \, \left({\bf \mathcal{\hat X}}_\tau,\mu_\tau^{\pi_{Cl}}\right)\,d\tau,
	\label{eq:index_Ob}
\end{equation} 
where ${U}_t^{Cl}$ is a cost function.
The resulting strategy will have a linear feedback form given by $\mu_t^{\pi_{Cl}}=\pi_{Cl}\, {\bf \mathcal{\hat X}}^{Cl}_t, \, \mu_t^{\pi_{Cl}}\in \mathbb{R}^{m}$.
\begin{assum}
	\label{Asm:2}
	There is a strategy ${\bf\pi}_{Cl}$ that is able to stabilize the closed-loop dynamics of the desired or approximated process $ \dot {\bf{\mathcal{\hat X}}}=({\bf \hat A+\hat B}\,{\bf\pi}_{Cl}){\bf {\bf{\mathcal{\hat X}}}}$.\hfill  $\square$
\end{assum}
Hence, the strategy ${\bf\pi}_{Cl}$ solves the underlying LQR problem of the desired or observed system~\eqref{eq:dyna}.
\subsection{Online Model-Following Strategy}  
The model-following strategy regulates the model-following errors $e^{Mf}_t$ between selected outputs of the process ${\bf Y}^s_t$ and those of the reference system ${\bf Y}^{ref}_t$ (i.e., $e^{Mf}_t={\bf \, Y}^{ref}-{\bf Y}^s_t$). Similar to the observer strategy, the model-following error samples are collected at a fixed-time interval $\delta$ such that ${\bf E}^{Mf}_t= \left[ \begin{array}{ccc}
e^{Mf}_t & e^{Mf}_{t+\delta} & e^{Mf}_{t+2\delta} \end{array}\right]^\textrm{T}  \in \mathbb{R}^{3q}$. The number of error samples is decided based on the expected order of the error dynamics. Three error samples will be considered for both the observer and model-following strategies. The model-following strategy ${\pi_{Mf}}$ will be decided online following the control law ${\bf u}^{\pi_{Mf}}_{t+\delta}={\bf u}^{\pi_{Mf}}_t+\mu_t^{\pi_{Mf}},$ where $\mu_t^{\pi_{Mf}}={{\bf \pi}_{Mf}}\,{\bf E}^{Mf}_t  \in \mathbb{R}^{m}$. A performance index is considered to evaluate the quality of ${\bf \pi}_{Cl}$ such that
\begin{equation}
J_t^{\bf \pi_{Mf}}=\int_t^\infty {U}_\tau^{Mf} \, \left({\bf E}^{Mf}_\tau,\mu_\tau^{\pi_{Mf}}\right)\,d\tau,
\label{eq:index_Mf}
\end{equation} 
where the model-following cost function is denoted by ${U}_t^{Mf}$.
\begin{assum}
\label{Asm:3}
The strategies ${\pi_{Cl}}$ and ${\pi_{Mf}}$ are able to stabilize the process around the desired reference-trajectory ${\bf Y}^{ref}$.  \hfill  $\square$
\end{assum}
\subsection{Overall Control Solution }  
The control mechanism implies existence of kernel solution structures that realize the interactive optimization goals of the sub-control problems (i.e., $\text{argmin}_{\pi_{Ob}}\ \, J_t^{\bf \pi_{Ob}},$ $\text{argmin}_{\pi_{Cl}} \, J_t^{\bf \pi_{Cl}},$ and $\text{argmin}_{\pi_{Mf}} \, J_t^{\bf \pi_{Mf}}$). The process is LTI system and hence the kernel solutions could take quadratic forms in the observer errors, observed states, and model-following errors. Assumptions \ref{Asm:1}, \ref{Asm:2}, and \ref{Asm:3} are made to ensure availability of such strategies that are able to stabilize the open-loop dynamics and follow the desired reference-trajectory. Moreover, this solution form can be attempted for nonlinear systems given the data-driven structure of its proposed strategies. The overall control scheme of the model-following solution is depicted by Fig.~\ref{fig:block}.    

\begin{figure}
	\centering
	\includegraphics[width=0.8\linewidth]{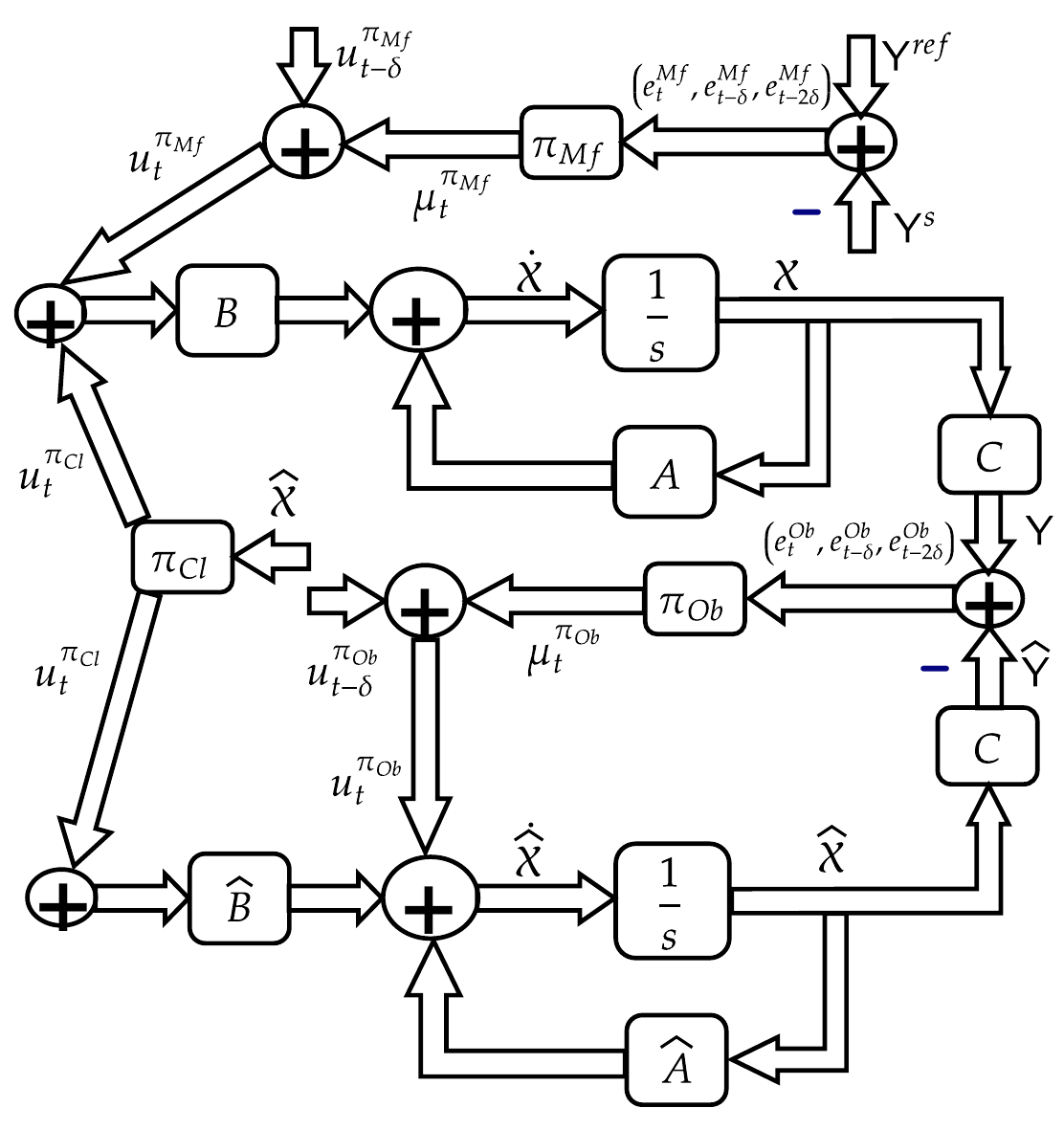}
	\caption{The overall control scheme}
	\label{fig:block}
\end{figure}

\section{Optimal Control Foundation\label{sec:OPT_IBE}}
This section provides an optimal control foundation that will be adopted by the adaptive learning schemes developed. The aim of each sub-control problem is to minimize the respective cost function given by
${U}_t^{i} \, \left(\bm{\mathcal{F}}^{i}_t,\mu_t^{\pi_{i}}\right)=\frac{1}{2}\left(\bm{\mathcal{F}}^{i\, T}_t\,{\bf Q}^i\,\bm{\mathcal{F}}^{i}_t+\mu_t^{\pi_{i \, T}}\,{\bf R}^i\,\mu_t^{\pi_{i}}\right),$ where $i \in \left\{Ob, Cl, Mf\right\} $ stands for a sub-control problem, ${\bf Q}^i \in \mathbb{R}^{n \times n},$ and ${\bf R}^i\in \mathbb{R}^{m \times m}$ are weighting matrices. Each optimization problem is solved using a Hamiltonian structure that is given by 
\begin{equation}
		\label{eq:Ham}
	H^i(\bm{\mathcal{F}}^i_t,\bm{\lambda}_t^{\pi_{i}}, \mu_t^{\pi_{i}})=\bm{\lambda}_t^{\pi_{i} \, T} \, \dot{\bm{Z}}^{\pi_i}_t+{U}_t^{i} \, \left(\bm{\mathcal{F}}^{i}_t,\mu_t^{\pi_{i}}\right),
\end{equation}  
where $	H^i, \, i \in \left\{ Ob,  Cl, Mf \right\}$ is a Hamiltonian function for each sub-control problem $i$, $\bm{\mathcal{F}}^{i}_t \,  \in  \, \left\{{\bf E}^{Ob}_t,\, {\bf \mathcal{\hat X}}^{Cl}_t, \, {\bf E}^{Mf}_t \, \right\}$,  and $\bm{\lambda}_t^{\pi_{i}} \in \mathbb{R}^{(n+m)} $ is a Lagrange multiplier associated with constraints $\dot{\bm{Z}}^{\pi_i}_t: \quad \bm{Z}^{\pi_i}={ \bm{\mathcal{F}}^{{i}\,T}_t , {{\bf \mu}^{\pi_i \, T}_t}}^T \, \in \mathbb{R}^{(n+m)}$. 

The next result shows how the kernel solution forms can be selected. This will be needed to ensure a rigorous temporal difference structure. 
\begin{lemma}
	\label{lem:Lyp}
	Let ${V}_t^i({\bm{Z}}^{\pi_i}_t) > 0, \, {V}_t^i(0)=0$ be a solving value function satisfying the Hamiltonian~\eqref{eq:Ham}. Then, ${V}_t^i({\bm{Z}}^{\pi_i}_t)$ represents a Lyapunov   Function.
\end{lemma} 

\begin{proof}%[Proof of Lemma 1]
The function ${V}_t^i$ makes use of the LTI dynamic properties of the underlying sub-control problems. Hence, its structure can be selected as
\begin{equation}
	J_t^{\bf \pi_{i}} {V}_t^i(\bm{Z}^{\pi_i}_t)
	=\frac{1}{2}
	{\bm{Z}}^{\pi_i}_t\,
	\bm{\mathcal{S}}^i\,
	{\bm{Z}}^{\pi_i \, T}_t,
	\label{val}
\end{equation}
where $\displaystyle  \bm{0}<
\bm{\mathcal{S}}^i
\equiv
\begin{bmatrix*}[l]
	\bm{\mathcal{S}}^i_{\bm{\mathcal{F}}\bm{\mathcal{F}}}
	& 
	\bm{\mathcal{S}}^i_{\bm{\mathcal{F}}{\bf \mu}^{\pi}}
	\\
	\bm{\mathcal{S}}^i_{{\bf \mu}^{\pi}\bm{\mathcal{F}}}
	&
	\bm{\mathcal{S}}^i_{{\bf \mu}^{\pi}{\bf \mu}^{\pi}}
\end{bmatrix*}
\in \mathbb{R}^{4\times4} ,
$
$\boldsymbol{S}^i_{\bm{\mathcal{F}}\bm{\mathcal{F}}} \in \mathbb{R}^{3 \times 3},$ and $\boldsymbol{S}^i_{{\bf \mu}^{\pi}{\bf \mu}^{\pi}} \in \mathbb{R}$.

This form represents a candidate Lyapunov function under the given assumptions. Further, the Hamilton-Jacobi (HJ) theory provides the relation between value function ${V}_t^i$ and Lagrange multiplier $\bm{\lambda}_t^{\pi_i}$ such that $\bm{\lambda}_t^{\pi_i}=\grad{{V}_t^i}={\partial {V}_t^i}/{\partial {\bm{Z}}^{\pi_i}_t}$. Moreover, the solution of the underlying optimal sub-control problem yields a solution for each Bellman equation (i.e., $	H^i(\bm{\mathcal{F}}^i_t,\grad{{V}_t^i}, \mu_t^{\pi_{i}})=0$) which implies that $ \frac{\partial {V}_t^i}{\partial {\bm{Z}}^{\pi_i}_t}^T \, \dot{\bm{Z}}^{\pi_i}_t+{U}_t^{i} \, \left(\bm{\mathcal{F}}^{i}_t,\mu_t^{\pi_{i}}\right)=0$.
This is an infinitesimal form of 
\begin{equation}
	\label{eq:HJBn}
{\dot V}_t^i+{U}_t^{i} \, \left(\bm{\mathcal{F}}^{i}_t,\mu_t^{\pi_{i}}\right)=0.
\end{equation}	
Since ${\dot V}_t^i\le0,$ then ${V}_t^i$ is a Lyapunov function.
\end{proof}

The realization of a model-free control strategy starts with building a Bellman-based optimization structure (i.e., a temporal difference equation). This structure can be adopted by different approximate heuristic dynamic programming forms. The following result explains how to do that.   

\begin{lemma}
	\label{lem:Bell}
	Let ${V}_t^{*i}({\bm{Z}}^{\pi_{*i}}_t) > 0, \, {V}_t^{*i}(0)=0$ be the optimal solution of the Hamiltonian~\eqref{eq:Ham} following the optimal strategy $\pi_{*i}$. Then, ${V}_t^{*i}({\bm{Z}}^{\pi_{*i}}_t)$ is the optimal solution of the Bellman optimality expression given by	
		\begin{equation}
		{V}_t^{*i}({\bm{Z}}^{\pi_{*i}}_t) =\int_t^{t+\delta} {U}_\tau^{*i} \, \left(\bm{\mathcal{F}}^{i}_\tau,\mu_\tau^{\pi_{*i}}\right)\, d\tau+{V}_t^{*i}({\bm{Z}}^{\pi_{*i}}_{t+\delta}).
		\label{eq:Bello}
	\end{equation}
\end{lemma} 

\begin{proof}%[Proof of Lemma 1]
The Hamiltonian~\eqref{eq:HJBn} can be written using Euler approximation so that
	\begin{equation}
	{V}_t^{i}({\bm{Z}}^{\pi_{i}}_t) =\int_t^{t+\delta} {U}_\tau^{i} \, \left(\bm{\mathcal{F}}^{i}_\tau,\mu_\tau^{\pi_{i}}\right)\, d\tau+{V}_t^{i}({\bm{Z}}^{\pi_{i}}_{t+\delta}).
	\label{eq:Bell}
\end{equation}	
The optimal strategy is calculated such that $\mu^{\pi_{*i}}_t=\argmin_{\mu^\pi}{V}_t^{i}({\bm{Z}}^{\pi_{i}}_t)$. Hence, it takes a linear form given by 
\begin{equation}
\mu_t^{\pi_{*i}}=- \bm{\mathcal{S}}^{i^{-1}}_{{\bf \mu}^{\pi}{\bf \mu}^{\pi}}\,	\bm{\mathcal{S}}^i_{{\bf \mu}^{\pi}\bm{\mathcal{F}}}\, \bm{\mathcal{F}}^{i}_t.
	\label{optpol}
\end{equation}	
This strategy results in an optimal function ${V}_t^{*i}({\bm{Z}}^{\pi_{*i}}_t)$. Therefore, the solution of $	H^i(\bm{\mathcal{F}}^i_t,\grad{{V}_t^{*i}}, \mu_t^{\pi_{*i}})=0$ (i.e., the Hamilton-Jacobi-Bellman (HJB) equation) is equivalent to solving the integral temporal difference equation~\eqref{eq:Bello}.   	
\end{proof}
The next result shows that, the observer-based model-following strategy is able to asymptotically stabilize the observer errors (i.e., $\lim\limits_{t \rightarrow \infty}\lVert{{\bf e}^{Ob}_t}\rVert\rightarrow \bf 0$) and the model-following errors (i.e., $\lim\limits_{t \rightarrow \infty}\lVert{{\bf e}^{Mf}_t}\rVert\rightarrow \bf 0$) as well. For clarity, the different tracking errors are refereed to as ${{\bf e}^{i}_t, \, i \in\left\{Ob,Mf\right\}},$. 
\begin{thm}
\label{thm:Equi}
Let the initial values of functions ${V}_0^{i}({\bm{Z}}^{\pi_{i}}_0), \forall i$ be bounded by upper values $\Upsilon^i, \forall i \in \left\{Ob,Mf\right\}$. Then, the trajectory-tracking dynamical error systems are asymptotically stable (i.e., $\lim\limits_{t \rightarrow \infty}\lVert{{\bf e}^{i}_t}\rVert\rightarrow \bf 0$).
\end{thm}

\begin{proof}
The integral Bellman equation~\eqref{eq:Bell} yields a Lyapunov function, as per Lemma~\ref{lem:Lyp}. Hence, ${V}_t^{i}({\bm{Z}}^{\pi_{i}}_t)\le {V}_0^{i}({\bm{Z}}^{\pi_{i}}_0)\le \Upsilon^i$ and ${V}_t^{i}({\bm{Z}}^{\pi_{i}}_t) \in \mathcal{L}_\infty, \forall i \in \left\{Ob,Mf\right\}$. This, Assumption \ref{Asm:1}, and Assumption \ref{Asm:3} reveal that the trajectory-tracking errors $\left\{e^{i}_t, e^{i}_{t+\delta}, e^{i}_{t+2\delta}\right\}\in \mathcal{L}_\infty$ and hence the stabilizing strategy, derived using the kernel solution $\bm{\mathcal{S}}^i$, is $\pi_i \in \mathcal{L}_\infty$. The HJB equation $	H^i(\bm{\mathcal{F}}^i_t,\grad{{V}_t^i}, \mu_t^{\pi_{i}})=0$ signifies that ${\dot V}_t^i=-{U}_t^{i} \, \left(\bm{\mathcal{F}}^{i}_t,\mu_t^{\pi_{i}}\right)\le0$. Therefore, ${\dot V}_t^i \in \mathcal{L}_\infty$ and $\dot {\bm{Z}}^{\pi_{i}}_t  \in \mathcal{L}_\infty$. This HJB equation yields  $\lim\limits_{t \rightarrow \infty}\lVert{{V}^{i}_t}\rVert\rightarrow \bf 0$ with $\int_0^t\frac{1}{2}\left(\bm{\mathcal{F}}^{i\, T}_\tau\,{\bf Q}^i\,\bm{\mathcal{F}}^{i}_\tau+\mu_\tau^{\pi_{*i \, T}}\,{\bf R}^i\,\mu_\tau^{\pi_{*i}}\right)\, d\tau \le {V}_0^{i}({\bm{Z}}^{\pi_{i}}_0)$. Then, $\int_0^t\frac{1}{2}\bm{\mathcal{F}}^{i\, T}_\tau\,\left({\bf Q}^i+{\pi^T_{*i}}\,{\bf R}^i\,{\pi_{*i}}\right)\,\bm{\mathcal{F}}^{i}_\tau\, d\tau \le {V}_0^{i}({\bm{Z}}^{\pi_{i}}_0)$. This reveals that $\bm{\mathcal{F}}_t^{i} \in \mathcal{L}_2$ and ${\dot V}_t^i \in \mathcal{L}_2$. Therefore, according to Barbalat's Lemma $\lim\limits_{t \rightarrow \infty}\, {\dot V}_t^i \, \rightarrow \bf 0,$ which means that the model-following and observation errors are stabilized asymptotically. 
\end{proof}

\section{IRL Solution Algorithm\label{sec:IRL_Sol}}
The analytical solution of the coupled integral Bellman optimality equations~\eqref{eq:Bello} cannot be done in a straight forward manner. Hence, approximate learning mechanisms are needed to solve these equations such as RL. As such, a model-free IRL solution is developed to find the best strategies-to-follow. 
\subsection{Integral Reinforcement Learning Algorithm}
Algorithm~\ref{Alg:Alg1} lays out the architecture of an online IRL solution. It solves the integral temporal difference equations~\eqref{eq:Bello} following the optimal strategies~\eqref{optpol} to solve the underlying sub-control problems simultaneously. This is done in a model-free fashion using the error measurements $e^{Mf}_t$ and $e^{Ob}_t$ in addition to the observed states ${\bf \mathcal{\hat X}}$. 
	\begin{algorithm}[htb!]
	\caption{\label{Alg:Alg1} Integral Reinforcement Learning Algorithm}
	\begin{enumerate}
		\item[{\footnotesize{}1:}] Initialize the states $\bm{\mathcal{F}}^{i}_0, \forall i$ and strategies ${\pi_{i}}, \forall i$ and hence compute the control signals $\mu_0^{\pi_{i}}, \forall i$.
		\item[{\footnotesize{}2:}] Calculate updated matrices $\bm{\mathcal{S}}^{i(r+1)}, \forall i$ by solving the integral Bellman equations 
		\begin{equation}
			\small
	{V}_t^{i(r+1)}({\bm{Z}}^{\pi_{i(r)}}_t)-{V}_t^{i(r+1)}({\bm{Z}}^{\pi_{i(r)}}_{t+\delta}) =\int_t^{t+\delta} {U}_\tau^{i(r)}  \left(\bm{\mathcal{F}}^{i}_\tau,\mu_\tau^{\pi_{i(r)}}\right)\, d\tau,
			\label{eq:PIVal}
		\end{equation}
		where $r$ is an iterative index.
		\item[{\footnotesize{}3:}] According to the updated solution matrices $\bm{\mathcal{S}}^{i(r+1)}$, find  
		\begin{equation}
\mu_{t+\delta}^{\pi_i^00{(r+1)}}=- \bm{\mathcal{S}}^{i(r+1)^{-1}}_{{\bf \mu}^{\pi}{\bf \mu}^{\pi}}\,	\bm{\mathcal{S}}^{i(r+1)}_{{\bf \mu}^{\pi}\bm{\mathcal{F}}}\, \bm{\mathcal{F}}^{i}_{t+\delta}.
			\label{eq:PIPol}
		\end{equation}
		\item[{\footnotesize{}4:}] Upon the convergence of $\lVert\bm{\mathcal{S}}^{i(r+1)}-\bm{\mathcal{S}}^{i(r)}\rVert, \, \forall i$ terminate the adaptation process.
	\end{enumerate}
\end{algorithm}
This algorithm is implemented in a model-free fashion using the means of adaptive critic. The actor structure approximates the underlying optimal strategy while the critic observes the quality of the attempted strategy.  
\subsection{Actor-Critic Implementation}  
A real-time parameter estimation approach is needed to solve for the best strategies-to-follow and this is done using two steps. The first considers structures $\bm{\hat \mathcal{S}}^i,\forall i$ to approximate matrices $\bm{\mathcal{S}}^i, \forall i$ (i.e., solving the underlying integral Bellman optimality equations) and the second approximates the optimal strategies $\pi_{*i}, \forall i$ taking the form of $\hat \pi_{i}, \forall i$. 
The vector-indices of the process are defined as follows $n=3$ and $m=p=s=1$. Hence, each function ${V}_t^i({\bm{Z}}^{\pi_i}_t)$ is approximated using the following critic structure 

\begin{equation}
	\hat{V}_t^i({\bm{Z}}^{\hat \pi_i}_t) =\frac{1}{2}{\bm{Z}}^{\hat \pi_i\, T}_t\, 	\bm{\hat \mathcal{S}}^i \, {\bm{Z}}^{\hat \pi_i}_t, \forall i
	\label{crit}
\end{equation}
where ${\hat \pi_i}^T \in \mathbb{R}^{3}$ and ${\bm 0}<\bm{\hat \mathcal{S}}^{i\, T} \in \mathbb{R}^{4\times4}$ are the weights of the actor and critic structures, respectively. 

The different integral Bellman optimality equations can be written as $\tilde{V}_{t,t+\delta}^i({\bm{Z}}^{\hat \pi_i}_{t,t+\delta}) =\int_t^{t+\delta} {U}_\tau^{*i} \, \left(\bm{\mathcal{F}}^{i}_\tau,\mu_\tau^{\pi_{*i}}\right)\, d\tau, \, \forall i$ with $\tilde{V}_{t,t+\delta}^i({\bm{Z}}^{\hat \pi_i}_{t,t+\delta})=\hat{V}_t^i({\bm{Z}}^{\hat \pi_i}_t)-\hat{V}_{t+\delta}^i({\bm{Z}}^{\hat \pi_i}_{t+\delta}), \forall i$. These equations can be reshaped such that 
	\begin{equation}
	{\bf\Theta}^i \tilde {\bm{Z}}^{\hat \pi_i}_t={\bf\Phi}^i_t,
	\label{crit}
\end{equation}
where  $\tilde {\bm{Z}}^{\hat \pi_i}_t=\left\{\left({\bf Z}^{\zeta_i}_{t} \bigotimes {\bf Z}^{\eta_i}_{t}\right), \, i\in \left\{Ob,Cl, Mf \right\}, \right.$ $\zeta_{Ob}=1,\dots,(\ell \times p+m), $ $\zeta_{Cl}=1,\dots,(n \times p+m), $ $\zeta_{Ob}=1,\dots,(\nu \times s+m), $
$\eta_{Ob}=\zeta_{Ob},\dots,(\ell \times p+m),$
$\eta_{Cl}=\zeta_{Cl},\dots,(n \times p+m),$
$\left. \eta_{Ob}=\zeta_{Ob},\dots,(\nu \times s+m)\right\},$ 
${\bf\Theta}^i$ is a vector that is calculated by reshaping matrix $\frac{1}{2}\bm{\hat \mathcal{S}}^i$ to associate its entries with $\tilde {\bm{Z}}^{\hat \pi_i}_t,$ and ${\bf\Phi}^i_t=\int_t^{t+\delta} {U}_\tau^{*i} \, \left(\bm{\mathcal{F}}^{i}_\tau,\mu_\tau^{\pi_{*i}}\right)\, d\tau.$

Similarly, the best strategy-to-follow is represented by \eqref{eq:PIPol}. Hence, each strategy is approximated by an actor structure $\hat \pi_i$ such that
\begin{equation}
	{\bm\hat \pi}_i \, \bm{\mathcal{F}}^i_t={\bm\phi}^i_t,
	\label{act}
\end{equation}
where vector ${\bm \hat \pi}_i$ represents the actor weights and ${\bm\phi}^i_t=- \bm{\mathcal{\hat S}}^{i^{-1}}_{{\bf \mu}^{\hat \pi}{\bf \mu}^{\hat \pi}}\,	\bm{\mathcal{\hat S}}^{i}_{{\bf \mu}^{\hat \pi}\bm{\mathcal{F}}^i} \, \bm{\mathcal{F}}^i_t$. 

The structures~\eqref{crit}~and~\eqref{act} are motivated by Algorithm~\ref{Alg:Alg1} and these forms highlight the coupling between the weights of the critic and actor structures. The weights will be tuned using parameter estimation approaches such are those rely on projection. The following result explains the convergence characteristics of the projection adaptation approach.

\begin{thm}
\label{thm:act-crt}
Let the actor weights $\hat \pi_i$ and critic weights $\bm{\hat \mathcal{S}}^i$ be calculated using Algorithm~\ref{Alg:Alg1}. Then, 
\begin{enumerate}
	\item[a.] The actor and critic weights converge to a set of weights $\hat \pi^*_i$ and $\bm{\hat \mathcal{S}}^{*i}$ found by solving \eqref{eq:PIVal} and \eqref{eq:PIPol} simultaneously.
	\item[b.] The deviations in actor and critic weights from the optimal solution (i.e., $\hat \pi^*_i$ and ${\bf\Theta}^{i*}$) are bounded under mild conditions about the adaptation paces.
\end{enumerate}
\end{thm}

\begin{proof}
\noindent 	
\textit{a.} The tuning errors in the adapted critic and actor weights are optimized using the Hamiltonian functions $H^i_{\bf\Theta}$ and $H^i_{\hat \pi}$, respectively. This is done along the trajectories of~\eqref{crit} and \eqref{act} as follows 
\begin{eqnarray*}
 H^i_{\bf\Theta}({\bf\Theta}^i,\lambda^i_{\bf\Theta}, f^i_{\bf\Theta})&=&\frac{1}{2}\left({\bf\Theta}^{i(r+1)}-{\bf\Theta}^{i(r)}\right)\left({\bf\Theta}^{i(r+1)}-{\bf\Theta}^{i(r)}\right)^\textrm{T}\\&+&\lambda^i_{\bf\Theta} f^i_{\bf\Theta} 	\\
  H^i_{\hat \pi}(  {\bm \hat \pi_i},  \lambda^i_{\bm \hat \pi}, f^i_{\bm \hat \pi})&=&\frac{1}{2}\left({\bm \hat \pi}^{r+1}_{i}-{\bm \hat \pi}^{r}_{i}\right)\left({\bm \hat \pi}^{r+1}_{i}-{\bm \hat \pi}^{r}_{i}\right)^\textrm{T}+\lambda^i_{\bm \hat \pi}f^i_{\bm \hat \pi} ,
 \end{eqnarray*}
where $\lambda^i_{\bf\Theta}$ and $\lambda^i_{\bm \hat \pi}$ are Lagrange multipliers associated with the optimization constraints ${f}^i_{\bf\Theta} ={\bf\Theta}^{i(r+1)} \tilde {\bm{Z}}^{\bm \hat \pi_i}_{r}-{\bf\Phi}^i_{r}$ and ${f}^i_{\bm \hat \pi}={\bm \hat \pi}_{i}^{r+1} \, {\bf \mathcal{F}}^i_r-{\bm \phi}^i_r$, respectively. The solution is achieved in real-time so that time-index $t$ is related to index $r$. This means that, the weights evaluated at index $(r+2)$ refer to calculations done at $t+2\delta$, for example.

Algorithm~\ref{Alg:Alg1} solves for the critic weights which are used to update the weights of the actor in that specific order. To find the critic and actor adaptation laws, the Hamiltonian optimization conditions are applied such that $\displaystyle \frac{\partial  H^i_{\bf\Theta}}{\partial {\bf\Theta}^{i(r+1)}}=0,\frac{\partial  H^i_{\bf\Theta}}{\lambda^i_{\bf\Theta}}=0, \frac{\partial  H^i_{\hat \pi}}{{\bm \hat \pi}^{r+1}_{i}}=0, \, \text{and} \, \frac{\partial  H^i_{\hat \pi}}{\lambda^i_{\bm \hat \pi}}=0$. Then, the following equations hold
\begin{eqnarray}
&\left({\bf\Theta}^{i(r+1)}-{\bf\Theta}^{i(r)}\right)^\textrm{T}+\lambda^i_{\bf\Theta}{\bm{Z}}^{\hat \pi_i}_t=0,	&{f}^i_{\bf\Theta}=0
\label{eq:ucrt}\\
	& \left({\bm \hat \pi}^{r+1}_{i}-{\bm \hat \pi}^{r}_{i}\right)^\textrm{T}+\lambda^i_{\bm \hat \pi} \bm{\mathcal{F}}^i_t =0,
	&f^i_{\bm \hat \pi}=0
	\label{eq:uact}
\end{eqnarray}
Manipulating~\eqref{eq:ucrt} and ~\eqref{eq:uact} yields the critic and actor update laws such that
\begin{eqnarray*}
{\bf\Theta}^{i(r+1)}&=&{\bf\Theta}^{i(r)}	-\frac{\tilde {\bm{Z}}^{\hat \pi_i \, \textrm{T}}_r}{\tilde {\bm{Z}}^{\hat \pi_i \, \textrm{T}}_r \, \tilde {\bm{Z}}^{\hat \pi_i}_r}\left({\bf\Theta}^{i(r)} \tilde {\bm{Z}}^{\bm \hat \pi_i}_{r}-{\bf\Phi}^i_{r}\right)
\\
 {\bm \hat \pi}^{r+1}_{i}&=&{\bm \hat \pi}^{r}_{i}-  \frac{\bm{\mathcal{F}}^{i \, \textrm{T}}_t}{\bm{\mathcal{F}}^{i\, \textrm{T}}_t \, \bm{\mathcal{F}}^i_t} \left({\bm \hat \pi}_{i}^{r} \, {\bf \mathcal{F}}^i_r-{\bm \phi}^i_r\right).
 \end{eqnarray*}
These actor-critic adaptation forms can be modified without violating the overall optimization objectives. This is done by controlling the adaptation paces and ensuring non-divergent behavior due to possible singularity issues such that
\begin{eqnarray}
	{\bf\Theta}^{i(r+1)}&=&{\bf\Theta}^{i(r)}	-\frac{\sigma^i_c \, \tilde {\bm{Z}}^{\hat \pi_i \, \textrm{T}}_r}{\alpha^i_c+\tilde {\bm{Z}}^{\hat \pi_i \, \textrm{T}}_r \, \tilde {\bm{Z}}^{\hat \pi_i}_r}\left({\bf\Theta}^{i(r)} \tilde {\bm{Z}}^{\bm \hat \pi_i}_{r}-{\bf\Phi}^i_{r}\right)
	\label{eq:crte}
	\\
	{\bm \hat \pi}^{r+1}_{i}&=&{\bm \hat \pi}^{r}_{i}-  \frac{\sigma^i_c \, \bm{\mathcal{F}}^{i \, \textrm{T}}_t}{\alpha^i_a+\bm{\mathcal{F}}^{i\, \textrm{T}}_t \, \bm{\mathcal{F}}^i_t} \left({\bm \hat \pi}_{i}^{r} \, {\bf \mathcal{F}}^i_r-{\bm \phi}^i_r\right),
	\label{eq:acte}
\end{eqnarray}
where $\sigma^i_c,\alpha^i_c,\sigma^i_a,$ and $\alpha^i_a, \, \forall i \in \mathbb{R}$ are positive parameters. The limits on their values will be explained later on.

Algorithm~\ref{Alg:Alg1} and the stability results highlighted by Lemma~\ref{lem:Bell} and Theorem~\ref{thm:Equi} reveal that
$\lim\limits_{t \rightarrow \infty}\lVert{{\bf \mathcal{F}}^i_t}\rVert\rightarrow \bf 0$
and
$\lim\limits_{t \rightarrow \infty}\lVert{\tilde {\bm{Z}}^{\bm \hat \pi_i}_{t}}\rVert\rightarrow \bf 0$. Therefore, the weights $\bm{\hat \mathcal{S}}^{i(r)}$ represented by ${\bf\Theta}^{i(r)}$ and ${\bm \hat \pi}^{r}_{i}$ will converge to a solution comprising a set of weights $\hat \pi^*_i$ and ${\bf\Theta}^{i*}$ or equivalently $\bm{\hat \mathcal{S}}^{*i}$, respectively.
\newline 

b. Let the adaptation errors in the updated critic and actor weights be given by ${\bf\Theta}_e^{i(r)}={\bf\Theta}^{i*}-{\bf\Theta}^{i(r)}$ and ${\bm \hat \pi}^{r}_{i \, e}={\bm \hat \pi}^{r}_{i}-{\bm \hat \pi}^{*}_{i},$ respectively. Then,~\eqref{eq:crte} yields $\displaystyle  {\bf\Theta}^{i(r+1) \, \textrm{T}}_e={\bf\Theta}^{i(r) \, \textrm{T}}_e-\frac{\sigma^i_c \,\tilde {\bm{Z}}^{\hat \pi_i}_r\, \left({\bf\Phi}^i_{r}-\left(-{\bf\Theta}^{i(r)}_e+{\bf\Theta}^{i*}\right) \tilde {\bm{Z}}^{\bm \hat \pi_i}_{r}\right)^{\textrm{T}}}{\alpha^i_c+\tilde {\bm{Z}}^{\hat \pi_i \, \textrm{T}}_r \, \tilde {\bm{Z}}^{\hat \pi_i}_r}$ with ${\bf\Phi}^i_{r}-{\bf\Theta}^{i*} \tilde {\bm{Z}}^{\bm \hat \pi_i}_{r}=0$. Then, $\displaystyle {\bf\Theta}^{i(r+1) \, \textrm{T}}_e={\bf\Theta}^{i(r) \, \textrm{T}}_e-\frac{\sigma^i_c \,\tilde{\bm{Z}}^{\hat \pi_i}_r\, \tilde{\bm{Z}}^{\hat \pi_i\, \textrm{T}}_r\,}{\alpha^i_c+\tilde {\bm{Z}}^{\hat \pi_i \, \textrm{T}}_r \, \tilde {\bm{Z}}^{\hat \pi_i}_r}{\bf\Theta}^{i(r) \, \textrm{T}}_e$ or simply $\displaystyle {\bf\Theta}^{i(r+1) \, \textrm{T}}_e= {\bm A}^i_c \, {\bf\Theta}^{i(r) \, \textrm{T}}_e,$ where ${\bm A}^i_c=\left(\textrm{I}_c^i-\frac{\sigma^i_c \,\tilde{\bm{Z}}^{\hat \pi_i}_r\, \tilde{\bm{Z}}^{\hat \pi_i\, \textrm{T}}_r\,}{\alpha^i_c+\tilde {\bm{Z}}^{\hat \pi_i \, \textrm{T}}_r \, \tilde {\bm{Z}}^{\hat \pi_i}_r}\right)$ and $\textrm{I}_c^i$ is an identity matrix. In order to ensure bounded tuning of the critic weights, the parameters $\sigma^i_c$ and $\alpha^i_c$ must be chosen such that $0<\sigma^i_c<2 \quad \text{and} \quad 0 < \alpha^i_c$. Similarly,~\eqref{eq:acte} leads to  ${\bm \hat \pi}^{r+1}_{i\, e}={\bm A}^i_a{\bm \hat \pi}^{r}_{i\, e}$ with $\displaystyle {\bm A}^i_a=\left(\textrm{I}_a^i-\frac{\sigma^i_a \,{\bf \mathcal{F}}^i_r{\bf \mathcal{F}}^{i \, \textrm{T}}_r}{\alpha^i_a+{\bf \mathcal{F}}^{i \, \, \textrm{T}}_r{\bf \mathcal{F}}^i_r}\right),$ where $\textrm{I}_a^i$ is an identity matrix. Therefore, the conditions  $0<\sigma^i_a<2 \quad \text{and} \quad 0 < \alpha^i_a$ are made to ensure proper convergence to solution $\hat \pi^*_i$.
\begin{rem}
The IRL solution is developed for Linear-Time-Invariant (LTI) systems. This setup can be adopted for nonlinear and time-varying systems due to its flexible data-driven and model-free structure, although no rigorous proofs were attempted yet beyond the current work. Herein, LTI system is considered to show the proof of concept. 
\end{rem}
\end{proof}
\section{Model-Following Validation Results\label{sec:simulation}}
The model-free IRL solution is validated using a third-order dynamical process. The parameters of the process are given by
$
\footnotesize{\bf A}=\left[
\begin{array}{ccc}
	0  &   1 &    0\\
	0  &  -5 &   10\\
	0  &  -1 &   -5
\end{array}
\right],$
$
{\bf B}=\left[
\begin{array}{c}
	0  \\ 0  \\	1  
\end{array}
\right],
$
and
$
{\bf C}=\left[
\begin{array}{ccc}
	0  & 1  & 0  
\end{array}
\right]^T
$. The dynamic parameters of the desired process are given by $
\footnotesize \hat {\bf A}=\left[
\begin{array}{ccc}
	0.0132 &   1.0085 &  -0.0055\\
	0.0132 &  -5.0286 &   9.9132\\
   -0.0526 &  -1.0155 &  -4.9374
\end{array}
\right]$ and $
\hat {\bf B}=\left[
\begin{array}{ccc}
   -0.0072 & -0.0547 &	1.0527
\end{array}
\right]^T$. The selected output of the process is given by ${\bf Y}^s=\mathcal{X}(2)$. Further, the nonlinear reference trajectory ${Y}^{ref}_{t}$ is given by 
\begin{align*}
{Y}^{ref}_{t}
& \small 	=
\begin{cases}
\displaystyle
1+\exp^{(-0.01 \, t)} \, \cos(\frac{1.5\, t}{20}), \hspace*{1em} \text{for }  t \le 10\,\\
\displaystyle
0.5\,\left(1+\exp(-0.01\,(t-10))\right),  \hspace*{1em} \text{for } 10  < t \le 20 \, 
\end{cases}
\end{align*}
The model-following goal is to regulate the errors ${Y}^{ref}_{t}-\mathcal{X}(2)$ such that $\lim\limits_{t \rightarrow \infty}\lVert{{Y}^{ref}_{t}-\mathcal{X}(2)}\rVert\rightarrow  0$. The remaining learning parameters are presented in Table~\ref{tab:Tab1}. The simulation is done using MATLAB software for $20 \, \sec$.
\begin{table}[htbp]
	\caption{Learning and Adaptation Parameters}
	\label{tab:Tab1}
	\centering
	%\begin{tabularx}{\textwidth}{@{}l*{11}{C}c@{}}
		\begin{tabular}{*{4}{p{1.5cm}}}
			\toprule
			Parameter   &  Value & Parameter   &  Value \\ 
			\midrule
			${\bf Q}^i$ & $0.05\,I_{3}$   & ${\bf R}^i$   & $0.01$ \\ $\delta$ &	$0.01 \, \sec$	& 
			$\alpha^i_c$    & $1.8$ \\  $\sigma^i_c$   & $0.5$&
			$\alpha^i_a$  & $1.8$  \\ $\sigma^i_a$   & $0.5$ &
			\\
			\bottomrule	
		\end{tabular}
	\end{table}

\paragraph*{Discussion} The desired dynamic response is having the form of an unstable process and the reference-trajectory is set to be nonlinear to challenge the performance of the IRL solution. 
Figures~\ref{fig:fig1} and \ref{fig:fig2} highlight the simulation results. The simulation plots related to the closed-loop, observer, and model-following loops adopt solid, dotted, and dotted-dashed lines, respectively for clarity. The critic and actor weights are computed using~\eqref{eq:crte}~and~\eqref{eq:acte}, respectively. The adaptations of the actor and critic weights are depicted by Figs.~\ref{fig:act} and \ref{fig:crt}, respectively. It is shown that after some exploration phase, the weights of the actor and critic structures converge to a solution for the underlying sub-control problem. The control signals $\mu_t^{\pi_{Ob}}$ and $\mu_t^{\pi_{Mf}}$ reveal the ability of the IRL solution to regulate the observation and model-following errors as implied by Figs.~\ref{fig:consig} and \ref{fig:err}, respectively. The closed-loop control signal $\mu_t^{\pi_{Cl}}$ allows the process to follow the desired reference-trajectory and optimize the closed-loop performance, as revealed by Fig~\ref{fig:consig}. The closed-loop strategy converges to ${\bf \pi}_{Cl}=\left[
\begin{array}{ccc}
  -15.9517 &  -4.0410 &  -4.9822
\end{array}
\right]$. Further, the eigenvalues of the exact and desired open-loop processes are given by $\left(0, \, -5 + 3.1623i, \, -5- 3.1623i\right)$ and $\left(0, \, -4.9764 + 3.1599i, \, -4.9764 - 3.1599i\right)$, respectively. While the eigenvalues of the closed-loop processes are given by $\left(-2.2139, \, -6.3842 + 5.5943i, \, -6.3842 - 5.5943i\right)$ and $\left(-2.1749, \, -6.3430 + 5.6979i,\, -6.3430 - 5.6979i\right),$ respectively. Although process~\eqref{eq:dyna} is unstable and the reference-trajectory is nonlinear, the observer and closed-loop strategies are able to stabilize the dynamic process while following the desired reference-trajectory. Overall, the control signal forms emphasized that the optimization goals are met. The performance of the process states and observed ones are shown in Fig.~\ref{fig:stat}. This plot highlights the ability of the IRL solution to achieve the control goals simultaneously in a model-free manner.

\section{Conclusion\label{sec:conclus}}
This work introduced a model-free integral reinforcement learning solution for a model-following control problem that is inspired by an observer scheme and a projection technique. The overall control structure included three interactive strategies. The solution was able to regulate the errors between the actual and observed states using a model-free strategy. This is done interactively with another model-free strategy that optimizes and stabilizes the closed-loop systems. Finally, a third strategy was considered to allow the process to follow a nonlinear reference-trajectory. The presented solution optimized the overall dynamic performance of the process while regulating the model-following errors. The implementation was done using an approximate projection approach to adapt the actor-critic weights comprising the different sub-strategies. The learning parameters were subject to mild conditions about the adaptation paces of the tuning laws to ensure convergence to a solution.
\newpage
	\begin{figure}[!hbt]
	\centering
	\subcaptionbox{Actor weights%
		\label{fig:act}}%[12em]%
	{%
		\includegraphics[width=0.45\textwidth]{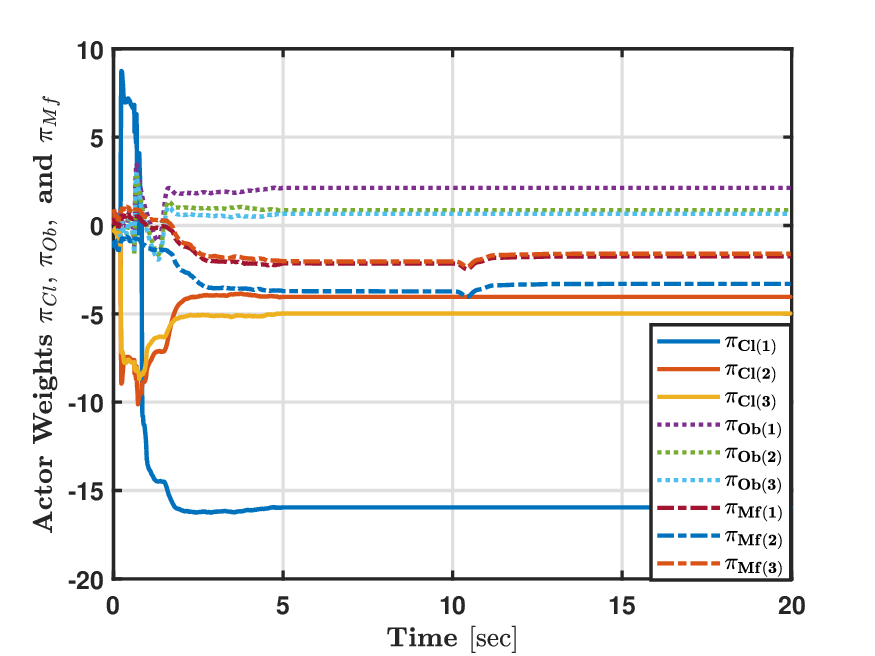}%
	}
	\\[1ex] 
	\subcaptionbox{Critic weights%
		\label{fig:crt}}%[12em]%
	{%
		\includegraphics[width=0.45\textwidth]{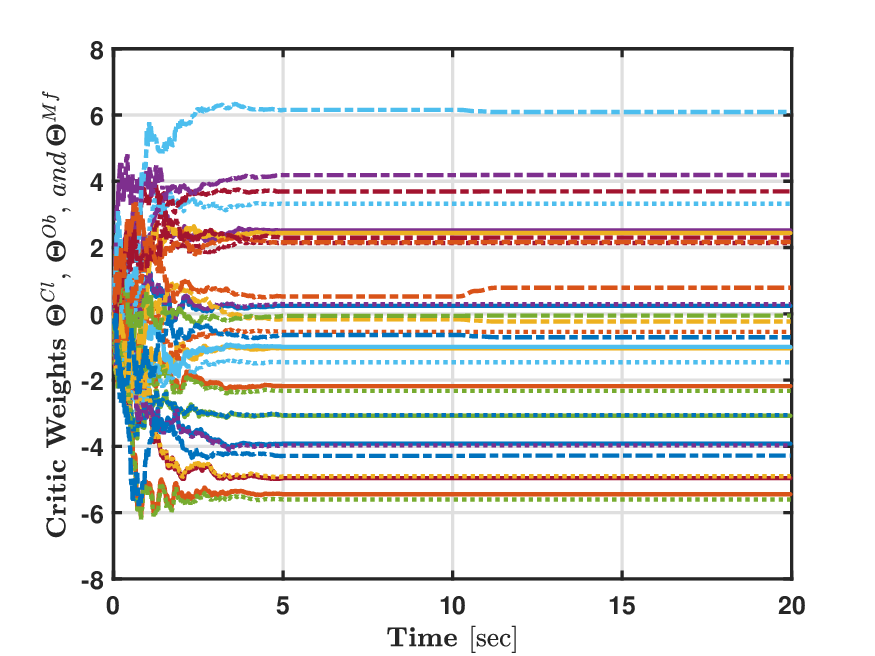}%
	}
	\caption{Adaptation of actor-critic weights  \label{fig:fig1}} 
\end{figure}

	\begin{figure}[!hbt]
	\centering
	\subcaptionbox{Control signals%
		\label{fig:consig}}%[12em]%
	{%
		\includegraphics[width=0.45\textwidth]{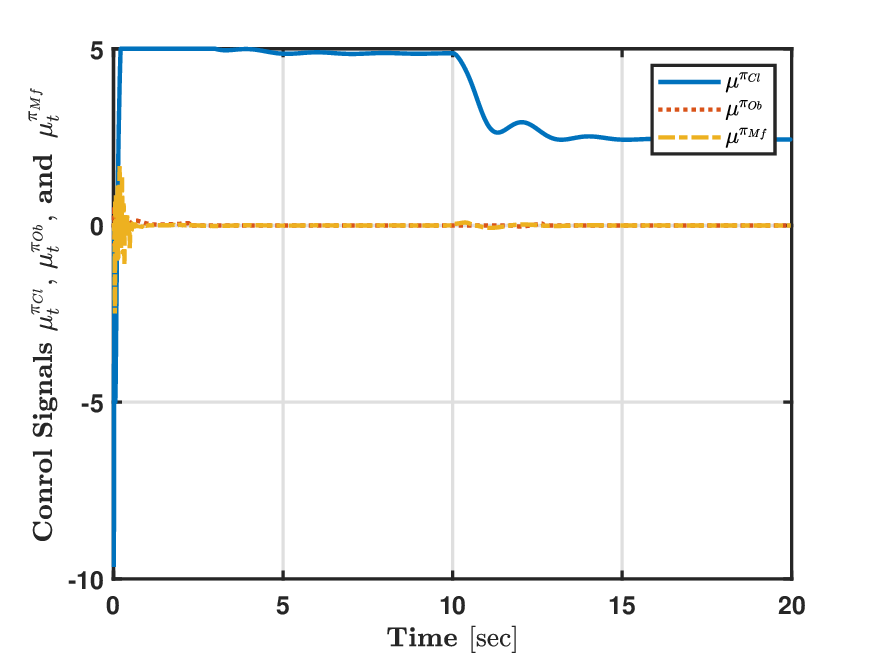}%
	}
	\\[1ex]
	\subcaptionbox{Model-following perfromance %
		\label{fig:stat}}%[12em]%
	{%
		\includegraphics[width=0.45\textwidth]{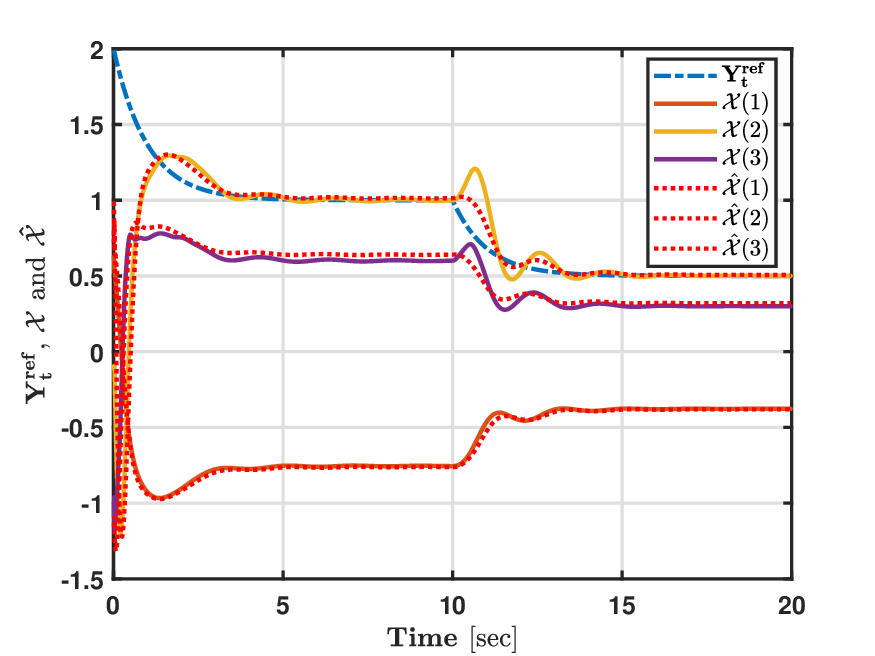}%
	}
	\\[1ex]
	\subcaptionbox{Observation and model-following errors%
		\label{fig:err}}%[12em]%
	{%
		\includegraphics[width=0.45\textwidth]{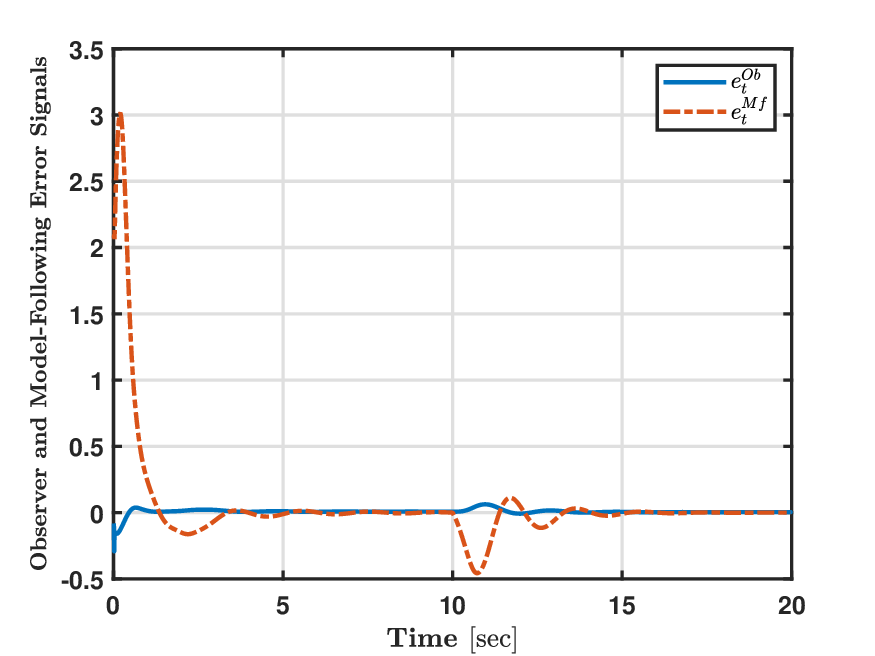}%
	}
	\caption{Performance of the overall control scheme \label{fig:fig2}} 
\end{figure}

\newpage

\balance

	\bibliographystyle{IEEEtran}
	\bibliography{bib_Neuro_Attit}

\begin{thebibliography}{10}
\providecommand{\url}[1]{#1}
\csname url@samestyle\endcsname
\providecommand{\newblock}{\relax}
\providecommand{\bibinfo}[2]{#2}
\providecommand{\BIBentrySTDinterwordspacing}{\spaceskip=0pt\relax}
\providecommand{\BIBentryALTinterwordstretchfactor}{4}
\providecommand{\BIBentryALTinterwordspacing}{\spaceskip=\fontdimen2\font plus
\BIBentryALTinterwordstretchfactor\fontdimen3\font minus
  \fontdimen4\font\relax}
\providecommand{\BIBforeignlanguage}[2]{{%
\expandafter\ifx\csname l@#1\endcsname\relax
\typeout{** WARNING: IEEEtran.bst: No hyphenation pattern has been}%
\typeout{** loaded for the language `#1'. Using the pattern for}%
\typeout{** the default language instead.}%
\else
\language=\csname l@#1\endcsname
\fi
#2}}
\providecommand{\BIBdecl}{\relax}
\BIBdecl

\bibitem{Lewis2012}
F.~L. Lewis, D.~Vrabie, and V.~L. Syrmos, \emph{Optimal Control}.\hskip 1em
  plus 0.5em minus 0.4em\relax John Wiley \& Sons, 2012.

\bibitem{aastrom2013adaptive}
K.~J. {\AA}str{\"o}m and B.~Wittenmark, \emph{Adaptive Control}.\hskip 1em plus
  0.5em minus 0.4em\relax Courier Corporation, 2013.

\bibitem{Byrne1995}
R.~Byrne and C.~Abdallah, ``Design of a model reference adaptive controller for
  vehicle road following,'' \emph{Mathematical and Computer Modelling},
  vol.~22, no.~4, pp. 343--354, 1995.

\bibitem{Kam1}
I.~Kaminer, A.~Pascoal, E.~Hallberg, and C.~Silvestre, ``Trajectory tracking
  for autonomous vehicles: An integrated approach to guidance and control,''
  \emph{Journal of Guidance, Control, and Dynamics}, vol.~21, no.~1, pp.
  29--38, 1998.

\bibitem{Liu2018}
J.~Liu, H.~An, Y.~Gao, C.~Wang, and L.~Wu, ``Adaptive control of hypersonic
  flight vehicles with limited angle-of-attack,'' \emph{IEEE/ASME Transactions
  on Mechatronics}, vol.~23, no.~2, pp. 883--894, 2018.

\bibitem{Chen2021}
H.~Chen, Y.~Peng, D.~Zhang, S.~Xie, and H.~Yan, ``Dynamic positioning for
  underactuated surface vessel via l1 adaptive backstepping control,''
  \emph{Transactions of the Institute of Measurement and Control}, vol.~43,
  no.~2, pp. 355--370, 2021.

\bibitem{MPC2021}
M.~Allenspach and G.~J.~J. Ducard, ``Nonlinear model predictive control and
  guidance for a propeller-tilting hybrid unmanned air vehicle,''
  \emph{Automatica}, vol. 132, p. 109790, 2021.

\bibitem{Hol2002}
W.~Dong, ``On trajectory and force tracking control of constrained mobile
  manipulators with parameter uncertainty,'' \emph{Automatica}, vol.~38, no.~9,
  pp. 1475--1484, 2002.

\bibitem{DMMPC2021}
M.~Bagherzadeh, S.~Savehshemshaki, and W.~Lucia, ``Guaranteed collision-free
  reference tracking in constrained multi unmanned vehicle systems,''
  \emph{IEEE Transactions on Automatic Control}, pp. 1--1, 2021.

\bibitem{AUV1}
R.~Cristi, F.~Papoulias, and A.~Healey, ``Adaptive sliding mode control of
  autonomous underwater vehicles in the dive plane,'' \emph{IEEE Journal of
  Oceanic Engineering}, vol.~15, no.~3, pp. 152--160, 1990.

\bibitem{Robust2019}
C.~Wu, A.~{van der Schaft}, and J.~Chen, ``Robust trajectory tracking for
  incrementally passive nonlinear systems,'' \emph{Automatica}, vol. 107, pp.
  595--599, 2019.

\bibitem{hashim2023observer}
H.~A. Hashim, A.~E. Eltoukhy, and A.~Odry, ``Observer-based controller for
  vtol-uavs tracking using direct vision-aided inertial navigation
  measurements,'' \emph{ISA transactions}, vol. 137, pp. 133--143, 2023.

\bibitem{hashim2023exponentially}
H.~A. Hashim, ``Exponentially stable observer-based controller for vtol-uavs
  without velocity measurements,'' \emph{International Journal of Control},
  vol.~96, no.~8, pp. 1946--1960, 2023.

\bibitem{Moore2014}
J.~Moore and R.~Tedrake, ``Adaptive control design for underactuated systems
  using sums-of-squares optimization,'' in \emph{2014 American Control
  Conference}, 2014, pp. 721--728.

\bibitem{MPC2014}
M.-C. Pai, ``Discrete-time sliding mode control for robust tracking and model
  following of systems with state and input delays,'' \emph{Nonlinear
  Dynamics}, vol.~76, no.~3, pp. 1769--1779, 2014.

\bibitem{SHI2017}
Z.~Shi and L.~Zhao, ``Robust model reference adaptive control based on linear
  matrix inequality,'' \emph{Aerospace Science and Technology}, vol.~66, pp.
  152--159, 2017.

\bibitem{netw1}
K.~Schenk, M.~Wissing, and J.~Lunze, ``Trajectory tracking in networks of
  linear systems,'' \emph{Automatica}, vol. 123, p. 109326, 2021.

\bibitem{AbouheafCTT2015}
M.~Abouheaf, F.~Lewis, M.~Mahmoud, and D.~Mikulski, ``Discrete-time dynamic
  graphical games: Model-free reinforcement learning solution,'' \emph{Control
  Theory and Technology}, vol.~13, no.~1, pp. 55--69, 2015.

\bibitem{AbouheafICRA19}
M.~Abouheaf and W.~Gueaieb, ``Multi-agent synchronization using online
  model-free action dependent dual heuristic dynamic programming approach,'' in
  \emph{2019 International Conference on Robotics and Automation (ICRA)}, 2019,
  pp. 2195--2201.

\bibitem{AbouheafAuto14}
M.~I. Abouheaf, F.~L. Lewis, K.~G. Vamvoudakis, S.~Haesaert, and R.~Babuska,
  ``Multi-agent discrete-time graphical games and reinforcement learning
  solutions,'' \emph{Automatica}, vol.~50, no.~12, pp. 3038--3053, 2014.

\bibitem{RMPC2021}
X.~Liu, L.~Ma, X.~Kong, and K.~Y. Lee, ``Robust model predictive iterative
  learning control for iteration-varying-reference batch processes,''
  \emph{IEEE Transactions on Systems, Man, and Cybernetics: Systems}, vol.~51,
  no.~7, pp. 4238--4250, 2021.

\bibitem{Obs2sen}
J.~Yang, W.~Tang, G.~Zhang, Y.~Sun, S.~Ademi, F.~Blaabjerg, and Q.~Zhu,
  ``Sensorless control of brushless doubly fed induction machine using a
  control winding current mras observer,'' \emph{IEEE Transactions on
  Industrial Electronics}, vol.~66, no.~1, pp. 728--738, 2019.

\bibitem{Obsv_2023}
W.~Li, X.~Chu, C.~Ma, and Y.~Kong, ``Finite-time model reference adaptive
  grasping control with fuzzy state observer for maglev grasping robot
  system,'' \emph{IEEE/ASME Transactions on Mechatronics}, pp. 1--12, 2023.

\bibitem{sut92}
R.~S. Sutton, A.~G. Barto, and R.~J. Williams, ``Reinforcement learning is
  direct adaptive optimal control,'' \emph{IEEE Control Systems Magazine},
  vol.~12, no.~2, pp. 19--22, 1992.

\bibitem{Sutton_1998}
R.~S. Sutton and A.~G. Barto, \emph{Reinforcement Learning: An Introduction},
  2nd~ed., ser. Second.\hskip 1em plus 0.5em minus 0.4em\relax Massachusetts:
  MIT Press, 1998.

\bibitem{Bertsekas1996}
D.~Bertsekas and J.~Tsitsiklis, \emph{Neuro-Dynamic Programming}, 1st~ed.\hskip
  1em plus 0.5em minus 0.4em\relax Massachusetts: Athena Scientific, 1996.

\bibitem{abouheaf2023online}
M.~I. Abouheaf, H.~A. Hashim, M.~A. Mayyas, and K.~G. Vamvoudakis, ``An online
  model-following projection mechanism using reinforcement learning,''
  \emph{IEEE Transactions on Automatic Control}, 2023.

\bibitem{AbouheafACC19}
M.~I. Abouheaf, M.~S. Mahmoud, and F.~L. Lewis, ``Policy iteration solution for
  differential games with constrained control policies,'' in \emph{2019
  American Control Conference (ACC)}, 2019, pp. 4301--4306.

\bibitem{Busoniu2010}
L.~Buşoniu, D.~Ernst, B.~De~Schutter, and R.~Babuška, ``Online least-squares
  policy iteration for reinforcement learning control,'' in \emph{Proceedings
  of the 2010 American Control Conference}, 2010, pp. 486--491.

\bibitem{Srivastava2019}
R.~Srivastava, R.~Lima, K.~Das, and A.~Maity, ``Least square policy iteration
  for ibvs based dynamic target tracking,'' in \emph{2019 International
  Conference on Unmanned Aircraft Systems (ICUAS)}, 2019, pp. 1089--1098.

\bibitem{IRL_Kyr}
K.~G. Vamvoudakis, D.~Vrabie, and F.~L. Lewis, ``Online adaptive algorithm for
  optimal control with integral reinforcement learning,'' \emph{International
  Journal of Robust and Nonlinear Control}, vol.~24, no.~17, pp. 2686--2710,
  2014.

\bibitem{Sutton2008}
S.~Bhatnagar, R.~S. Sutton, M.~Ghavamzadeh, and M.~Lee, ``Natural
  actor–critic algorithms,'' \emph{Automatica}, vol.~45, no.~11, pp.
  2471--2482, 2009.

\bibitem{CrtBahare}
B.~Kiumarsi and F.~L. Lewis, ``Actor–critic-based optimal tracking for
  partially unknown nonlinear discrete-time systems,'' \emph{IEEE Transactions
  on Neural Networks and Learning Systems}, vol.~26, no.~1, pp. 140--151, 2015.

\bibitem{CrtLewis}
R.~Song, F.~Lewis, Q.~Wei, H.-G. Zhang, Z.-P. Jiang, and D.~Levine, ``Multiple
  actor-critic structures for continuous-time optimal control using
  input-output data,'' \emph{IEEE Transactions on Neural Networks and Learning
  Systems}, vol.~26, no.~4, pp. 851--865, 2015.

\bibitem{CrtZhao}
X.~Zhao, S.~Han, B.~Tao, Z.-P. Yin, and H.~Ding, ``Model-based actor-critic
  learning of robotic impedance control in complex interactive environment,''
  \emph{IEEE Transactions on Industrial Electronics}, pp. 1--1, 2021.

\bibitem{Bah2017}
B.~Kiumarsi, F.~L. Lewis, and Z.-P. Jiang, ``$h_\infty$ control of linear
  discrete-time systems: Off-policy reinforcement learning,''
  \emph{Automatica}, vol.~78, pp. 144--152, 2017.

\bibitem{Lewis20}
Y.~Jiang, J.~Fan, W.~Gao, T.~Chai, and F.~L. Lewis, ``Cooperative adaptive
  optimal output regulation of nonlinear discrete-time multi-agent systems,''
  \emph{Automatica}, vol. 121, p. 109149, 2020.

\bibitem{MFAC2021}
J.~Xu, N.~Lin, and R.~Chi, ``Improved high-order model free adaptive control,''
  in \emph{2021 IEEE 10th Data Driven Control and Learning Systems Conference
  (DDCLS)}, 2021, pp. 704--708.

\bibitem{Bahare14}
B.~Kiumarsi, F.~L. Lewis, H.~Modares, A.~Karimpour, and M.-B. Naghibi-Sistani,
  ``Reinforcement q-learning for optimal tracking control of linear
  discrete-time systems with unknown dynamics,'' \emph{Automatica}, vol.~50,
  no.~4, pp. 1167--1175, 2014.

\bibitem{AbouheafTrans20}
M.~Abouheaf, N.~Q. Mailhot, W.~Gueaieb, and D.~Spinello, ``Guidance mechanism
  for flexible-wing aircraft using measurement-interfaced machine-learning
  platform,'' \emph{IEEE Transactions on Instrumentation and Measurement},
  vol.~69, no.~7, pp. 4637--4648, 2020.

\end{thebibliography}
	% name your BibTeX data base
\end{document}